\documentclass[10pt,twocolumn]{IEEEtran}

\IEEEoverridecommandlockouts

\usepackage{amsthm, amssymb}

\def\b{\mathbb}

\def\sinr{\textrm{SINR}}
\def\snr{\textrm{SNR}}

\def\dint{{\rm  d}}

\newtheoremstyle{slplain}
{3pt}
{3pt}
{\slshape}
{}
{\bfseries}
{.}%
{ }
{}

\theoremstyle{slplain}

\newtheorem{cor}{Corollary}

\newtheorem{pro}{Proposition}

\usepackage{cite}
\usepackage{amsfonts}
\usepackage{multicol,multienum}
\usepackage{subfigure}
\usepackage{multirow}

\ifCLASSINFOpdf
\usepackage[pdftex]{graphicx}
\graphicspath{{../pdf/}{../jpeg/}}
\DeclareGraphicsExtensions{.pdf,.jpeg,.png}
\else
\usepackage{float}
\usepackage[dvips]{graphicx}
\graphicspath{{../eps/}}
\DeclareGraphicsExtensions{.eps}
\fi

\usepackage[cmex10]{amsmath}
\usepackage{algorithm}
\usepackage{algorithmic}
\usepackage{array}
\usepackage{url}
\usepackage{setspace}

\begin{document}

\title{Performance Analysis of Asynchronous Multicarrier Wireless Networks}

\author{
\IEEEauthorblockA{Xingqin Lin, Libin Jiang, and Jeffrey G. Andrews
}
\thanks{
Xingqin Lin and Jeffrey G. Andrews are with the Department of Electrical $\&$ Computer Engineering, The University of Texas at Austin, TX, USA. (Email: xlin@utexas.edu, jandrews@ece.utexas.edu). Libin Jiang is with Qualcomm Flarion Tech., Bridgewater, NJ, USA. (Email: libinj@qti.qualcomm.com). 
Date revised: \today.
}
}

\maketitle

\begin{abstract}
This paper develops a novel analytical framework for \textit{asynchronous} wireless networks deploying multicarrier transmission. Nodes in the network have different notions of timing, so from the viewpoint of a typical receiver, the received signals from different transmitters are asynchronous, leading to a loss of orthogonality between subcarriers.  We first develop a detailed link-level analysis based on OFDM, based on which we propose a tractable system-level signal-to-interference-plus-noise ratio (SINR) model for asynchronous OFDM networks. The proposed model is used to analytically characterize several important statistics in asynchronous networks with spatially distributed transmitters, including (i) the number of decodable transmitters, (ii) the decoding probability of the nearest transmitter, and (iii) the system throughput. The system-level loss from lack of synchronization is quantified, and to mitigate the loss, we compare and discuss four possible solutions including extended cyclic prefix, advanced receiver timing, dynamic receiver timing positioning, and semi-static receiver timing positioning with multiple timing hypotheses.  The model and results are general, and apply to ad hoc networks, cellular systems, and neighbor discovery in device-to-device (D2D) networks.
\end{abstract}

\section{Introduction}

Consider a wireless network in which some nodes are broadcasting multicarrier signals and some nodes are listening. A listening node can decode the signal broadcast by a transmitting node if the received signal-to-interference-plus-noise ratio (SINR) exceeds some detection threshold, which depends on the used modulation and coding scheme. In reality, the various nodes in the network do not have precise synchronization with one another.   Therefore, this paper investigates the following question: if we take a snapshot of the network at a randomly selected time-frequency slot and randomly select a receiving node, then how many (if any) transmitting nodes can be decoded by the selected receiving node given that the network nodes each have different notions of timing? 

The general question posed above is of interest in many wireless networks. For example, in device-to-device (D2D) node discovery, a user equipment (UE) seeks to identify other UEs in its proximity via periodically broadcasting/receiving discovery signals  \cite{lin2013overview}. The number of transmitting devices that can be decoded is an important metric of discovery effectiveness.  A similar metric can be used for neighbor discovery in wireless ad hoc networks \cite{vasudevan2005neighbor}.  Cellular networks are a third important example. For example, in the downlink of a cellular network, how many base stations (BSs)  can be decoded by a typical UE at a given SINR?  This is a key consideration for soft handover or multiple-BS coverage or offloading in dense heterogeneous networks \cite{keeler2013sinr}.

\subsection{Background and Related Work}
\label{subsec:motivation}

The answer to the posed question obviously depends on how the transmitting nodes are spatially distributed.  We assume that the transmitting nodes are distributed according to a Poisson point process (PPP), which has two main advantages: (i) it captures the randomness inherent in the positions of the transmitting nodes (which are usually unknown to a receiving node), and (ii) the PPP has nice properties which make it particularly appealing from an analytical perspective \cite{baccelli2009stochastic}.  The PPP also has been recently shown to accurately model (with small modifications or shifts) a very large class of wireless networks, including even regular grids (with sufficient shadowing) \cite{blaszczyszyn2013using} and most random spatial distributions with a small and constant SINR shift \cite{GuoHae14}.  It is therefore reasonable to assume that the conclusions in this paper also will hold for most plausible network topologies.

Because of its excellent analytical properties, the PPP has found numerous applications in various types of wireless networks including the analysis and optimization of Aloha in wireless ad hoc networks \cite{baccelli2006aloha} and coverage and rate analysis in cellular systems \cite{andrews2011tractable}. More recently, the PPP has been applied to D2D networking including the analysis and design of scheduling in FlashlinQ \cite{wu2013flashlinq, baccelli2012optimizing}, the interaction between D2D and cellular systems \cite{erturk2013distributions, lin2013multicast, lin2013comprehensive}, and message dissemination with intermittently connected D2D links \cite{li2014message}. More applications of the PPP may be found in \cite{haenggi2009stochastic,elsawy2013stochastic} and references therein.

Despite this encouraging progress in applying the PPP to wireless networking, existing works nearly universally assume that the networks are perfectly synchronized. In cellular networks, BSs in different cells may not be synchronized in a Frequency Division Duplex (FDD) deployment, or have synchronization errors in a Time Division Duplex (TDD) deployment. These facts also lead to synchronization issues in D2D discovery. In particular, UEs participating in the discovery are synchronized with their associated BSs and thus may not be synchronized or at best imperfectly synchronized among themselves even when other factors like propagation delays are not considered, let alone the UEs that are out of cellular coverage \cite{lin2013overview}. The synchronous assumption becomes even more questionable when it comes to an ad hoc network in which network-wide synchronization is almost impossible. In such contexts, different transmitters have different notions of timing. From the viewpoint of a typical receiver, which also has its own notion of timing, the multicarrier orthogonal frequency-division multiplexing (OFDM) signals from the transmitters are asynchronous and also do not align with the receiver's timing, leading to a loss of orthogonality between subcarriers. 

The impact of synchronization errors on single-user OFDM has been extensively investigated in the literature (see e.g. \cite{pollet1994ber, schmidl1997robust, speth1999optimum, wang2003ser, mostofi2006mathematical}). Extension of the analysis in single-user OFDM to multiuser OFDM, however, is not straightforward as the latter involves a much larger set of random variables. Analysis of asynchronous OFDM in the uplink of cellular systems includes \cite{tonello2000analysis, el2001ofdm, park2003performance, raghunath2009sir}, while the downlink counterpart may be found in \cite{hamdi2009interference, medjahdi2011performance} and ad hoc networks in \cite{hamdi2010precise}. The works \cite{tonello2000analysis, el2001ofdm, park2003performance, raghunath2009sir} are focused on a single-cell setting and do not consider other-cell interference that plays a key role in system-level performance. In contrast, cochannel interference is modeled and studied in \cite{hamdi2009interference, medjahdi2011performance, hamdi2010precise}. But \cite{hamdi2009interference, medjahdi2011performance, hamdi2010precise} do not  consider or leverage the randomness inherent in the positions of network nodes, and the system-level studies therein are mainly based on Monte Carlo simulations.

\subsection{Main Results and Contributions}

The main goal of this paper is to incorporate the impact of asynchronous OFDM transmissions in the system-level study of wireless networks in which the positions of transmitting nodes are modeled by a PPP. The main contributions and outcomes of this paper are summarized as follows.

\subsubsection{A tractable SINR model for asynchronous OFDM networks} 

We carry out a detailed link-level analysis on the impact of timing misalignment in OFDM transmission. Based on the link-level analysis, we propose a tractable first-order SINR model, which can be conveniently used in system-level studies.

\subsubsection{System-level analysis of asynchronous PPP networks} 

We apply the proposed SINR model to study the system-level performance of asynchronous networks where the locations of transmitting nodes are modeled by a PPP and an OFDM waveform is used. Taking from a typical receiver's point of view, we derive analytical results for the average number of decodable transmitters, the decoding probability of the nearest transmitter, and system throughput. Further, we derive an upper bound on the distribution of the number of decodable transmitters. Note that, according to Palm theory \cite{baccelli2009stochastic}, the statistical performance experienced by a typical receiver is equivalently the spatially averaged performance over all receivers. 
The analysis of perfectly synchronized networks can be treated as a special case of this work. For example, the result on the decoding probability of the nearest transmitter reduces to \cite{andrews2011tractable} that studies a perfectly synchronized cellular network.

\subsubsection{Solutions for mitigating the impact of asynchronous transmissions} 

We compare and discuss four possible solutions including extended cyclic prefix, advanced receiver timing, dynamic receiver timing positioning, and semi-static receiver timing positioning with multiple timing hypotheses. These solutions, detailed in Section \ref{sec:solutions}, differ in complexity and may be applicable in different scenarios for mitigating the loss due to asynchronous transmissions.

The rest of this paper is organized as follows. Section \ref{sec:model} describes the system model. In Section \ref{sec:sinr}, we propose a tractable SINR model for asynchronous OFDM networks. System-level analysis is carried out in Sections \ref{sec:on}. Section \ref{sec:solutions} presents four possible solutions for mitigating the impact of asynchronous transmissions, and is followed by our concluding remarks in Section \ref{sec:conclusion}.
Numerical and/or simulation results are presented throughout the paper to help understand the various analytical results and build intuition.




\section{System Model}
\label{sec:model}

We consider a network in which transmitters use an OFDM waveform. The baseband equivalent time-domain signal $s_i (t)$ emitted by transmitter $i$ can be written as
\begin{align}
s_i (t) = & \sqrt{E_{i}} \sum_{m=-\infty}^{\infty} \frac{1}{N} \sum_{k }  S_i [k; m] \notag \\
&\times e^{ j 2 \pi \frac{k}{T} (t - mT_s) } \mathbb{I}_{[-T_{cp}, T_d)} ( t- mT_s) ,
\label{eq:111}
\end{align}
where $E_{i}$ denotes the transmit energy per sample of transmitter $i$, $m$ is the OFDM symbol index, $N$ denotes the total number of subcarriers, $k$ is the subcarrier index, $S_i[k;m]$ denotes transmitter $i$'s data symbol on the $k$-th subcarrier during the $m$-th OFDM symbol, $T_s = T_d +T_{cp}$ denotes the duration of an OFDM symbol with $T_d$ denoting the duration of the data part and $T_{cp}$ the duration of the cyclic prefix, and $\mathbb{I}_A (t)$ is an indicator function: it equals $1$ if $t\in A$ and zero otherwise. The data symbols $\{S_i [k; m]\}$ are complex and assumed to be independent and identically distributed (i.i.d.) with zero mean and unit variance.

As indicated in Section \ref{subsec:motivation}, we are interested in asynchronous scenarios where different transmitters have different notions of timing and so do the receivers. The more commonly studied synchronous scenarios where all the nodes are synchronized is a special case of this model. In an asynchronous network, we are interested in what a typical receiver ``sees'' at a random time-frequency resource unit. Note that the spectral width can be arbitrary. It can be a complete OFDM channel or a subband of an OFDM channel. In the latter case, transmitter $i$ simply puts zero-valued data symbols $S_i [k; m]$ on the unused subcarriers, as in OFDMA.

The active transmitters at the time-frequency resource unit in question  are assumed to be randomly distributed according to a PPP $\Phi$ with density $\lambda$. The location of transmitter $i \in \Phi$ is denoted by $X_i$. Note that our model does not preclude the possibility that there may be other transmitters active at some other time-frequency resource units. For example, we may consider a super PPP $\Phi' \supseteq \Phi$, where $\Phi'$ denotes the set of all the nodes in the network, and a time-frequency grid composed of orthogonal time-frequency resource units.\footnote{We ignore possible leakages from other time-frequency resource units when considering a particular time-frequency resource unit.} Each node randomly selects a time-frequency resource unit and transmits an OFDM waveform. Then the active transmitters at a randomly selected time-frequency resource unit constitute a PPP $\Phi$, thinned from the super PPP $\Phi'$. This described random access scheme is in fact part of the D2D discovery design used in LTE Direct \cite{3gppD2dRF}.

In this asynchronous network, we will study system-level questions such as the number of transmitting nodes that can be decoded by a typical receiver. To this end, since the transmitter process is stationary, we may assume without loss of generality that the typical receiver is located at the origin. Further, we consider flat-fading OFDM channels, i.e., the multipath spreads are small (w.r.t. sampling period). The last assumption holds for example in the following three scenarios: (1) there are not many obstacles in the radio environment and the arrival times of the multipaths are not resolvable at the receiver; (2) the received signal power is dominated by a single path, e.g. the line-of-sight path if it exists; and (3) the transmit signal is restricted to a flat-fading subband of a frequency-selective channel, as in OFDMA. We leave the important extension to frequency-selective OFDM channels as future work.

More specific modeling assumptions related to the system-level study will be given in Section \ref{sec:on}. 



\section{Tractable SINR Model for Asynchronous Networks}
\label{sec:sinr}

\subsection{Link-Level Timing Misalignment Analysis}

In this subsection, we analyze the impact of timing misalignment from a link-level perspective. Though similar analysis may be found in the rich OFDM literature (see e.g. \cite{speth1999optimum}), we briefly revisit this analysis to motivate our proposed SINR model that captures the impact of asynchronous transmission. To this end, we shall focus on the link between transmitter $i$ and the typical receiver and ignore the signals from the other transmitters for now.

Note that the $n$-th time-domain sample of the $m$-th OFDM symbol from the signal $s_i (t)$ is given by
\begin{align}
&s_i[n;m] = s_i \left(mT_s + n \frac{T_d}{N} \right) \notag \\
&=  \frac{\sqrt{E_{i}}}{N} \sum_{k} S_i [k; m] e^{ j 2 \pi \frac{k}{N} n } ,   n= -N_{cp},...,N-1,
\label{eq:13}
\end{align}
where $N_{cp} = N T_{cp} / T_d$ is the number of cyclic prefix samples.
Denote by $D_i$ the timing misalignment between transmitter $i$ and the typical receiver. Without loss of generality, we assume $D_i \in \mathbb{D} \triangleq [-(N+N_{cp}), N+N_{cp} )$.\footnote{This assumption can be easily relaxed by using different notations $m$ and $m'$ to respectively index OFDM symbols at the transmitter and at the receiver in the following analysis.}

In each OFDM symbol $m$, the typical receiver would like to decode the $m$-th OFDM symbol sent by transmitter $i$. To this end, it discards the first $N_{cp}$ samples falling in the current receiving window and performs a fast Fourier transform (FFT) on the remaining $N$ samples. We consider the following four cases, in which for notational simplicity we drop the additive noise term and assume that the channel gain is $1$ unless otherwise noted.

\textbf{Case 1: $-(N+N_{cp}) \leq D_i < -N$.} The $N$ samples used for the FFT of the $m$-th OFDM symbol are
\begin{align}
y[n;m] = s_i[n - D_i - N - N_{cp}; m+1],  n=0,...,N-1.
\label{eq:12}
\end{align}
The received signal on the $\ell$-th subcarrier during the $m$-th OFDM symbol is given by
\begin{align}
Y[\ell; m] 
&= \sqrt{E_{i}} e^{ j 2 \pi \frac{\ell}{N} (-D_i - N_{cp}) } S_i [\ell; m+1] ,
\label{eq:0}
\end{align}
which is derived in Appendix \ref{app:deri}. 
Thus, the received symbol on the $\ell$-th subcarrier during OFDM symbol time $m$ is just a phase rotated version of the transmitted symbol on the $\ell$-th subcarrier during OFDM symbol time $m+1$. If $S_i[\ell;m]$ is desired, the useful signal power is $0$. Otherwise, transmitter $i$'s signal appears as interference and its interference power (energy/symbol) on the $\ell$-th subcarrier during the $m$-th OFDM symbol equals
\begin{align}
P_{i}[\ell; m] = \mathbb{E} [ | Y[\ell; m] |^2 ] = G_i[m] E_{i} ,
\end{align}
where we have included the effect of channel gain $G_i[m]$ from transmitter $i$ to the typical receiver during OFDM symbol time $m$. Note that $G_i[m]$ is independent of subcarrier $\ell$ as we assume that the channel is flat-fading.

\textbf{Case 2: $-N \leq D_i < 0$.} The $N$ samples used for the FFT of  the $m$-th OFDM symbol are
\begin{align}
&y[n;m] = \notag \\
&\begin{cases}
s_i[-D_i+n; m], & 0\leq n \leq N-1 + D_i ; \\
s_i[n - (N + D_i) - N_{cp}; m+1], & N + D_i \leq n \leq N-1 .
\end{cases}
\label{eq:14}
\end{align}
The received signal on the $\ell$-th subcarrier during the $m$-th OFDM symbol is given by
\begin{align}
&Y[\ell; m] 
= \sqrt{E_{i}}\frac{N + D_i}{N} S_i[\ell;m] e^{- j 2 \pi \frac{\ell}{N} D_i } \notag \\ 
&- \sqrt{E_{i}}\frac{D_i}{N} S_i[\ell;m+1] e^{ j 2 \pi \frac{\ell}{N} (-D_i-N_{cp}) }  \notag \\
&+ \sqrt{E_{i}} \frac{1}{N} \sum_{k \neq \ell} \left( \frac{1- e^{ j 2 \pi \frac{k-\ell}{N}(N + D_i) }}{1- e^{ j 2 \pi \frac{k-\ell}{N} }} \right) \notag \\
&\times \left( S_i [k; m] e^{- j 2 \pi \frac{k}{N} D_i } - S_i [k; m+1] e^{ j 2 \pi \frac{k}{N} (-D_i-N_{cp}) } \right), 
\label{eq:1}
\end{align}
which is derived in Appendix \ref{app:deri}. Thus, the total received power on the $\ell$-th subcarrier during the $m$-th OFDM symbol from transmitter $i$ is
\begin{align}
P_{i}[\ell; m] 
=& G_i[m]E_{i} \bigg( \frac{1}{N^2} \left((N + D_i)^2 + D_i^2\right) \notag \\
&+ \frac{2}{N^2} \sum_{k \neq \ell} \frac{ \sin^2 \left( \pi \frac{N + D_i}{N} (k-\ell) \right)}{ \sin^2 \left( \pi \frac{1}{N} (k-\ell) \right) } \bigg),
\label{eq:2}
\end{align}
where we have used the assumption that $\{S_i [k; m]\}$ are i.i.d. and have zero mean and unit variance.
If $S_i[\ell;m]$ is desired, the useful signal power is $\frac{(N + D_i)^2 }{N^2} G_i[m]E_{i}$; the remaining terms in (\ref{eq:2}) contribute to self-interference including both inter-carrier interference (ICI) and inter-symbol interference (ISI). Otherwise, transmitter $i$'s signal appears as interference whose power is characterized by (\ref{eq:2}).

\textbf{Case 3: $0 \leq D_i < N_{cp}$.} The $N$ samples used for the FFT of  the $m$-th OFDM symbol are
\begin{align}
y[n;m] =
s_i[n - D_i; m], \quad & 0 \leq n \leq N-1.
\end{align}
As in Case 1, we can show that the received signal on the $\ell$-th subcarrier during the $m$-th OFDM symbol is given by
\begin{align}
Y[\ell; m] 
&= \sqrt{E_{i}} S_i [\ell; m] e^{- j 2 \pi \frac{\ell}{N} D_i } .
\end{align}
If $S_i[\ell;m]$ is desired, the useful signal power is $G_i[m] E_{i} $, and there is no self-interference. Otherwise, transmitter $i$'s signal appears as interference with power $G_i[m] E_{i} $.

\textbf{Case 4: $N_{cp} \leq D_i <  N+N_{cp}$.} The $N$ samples used for the FFT of the $m$-th OFDM symbol are
\begin{align}
&y[n;m] = \notag \\
&\begin{cases}
s_i[n + N + N_{cp} - D_i ; m-1], & 0\leq n \leq D_i - N_{cp} -1; \\
s_i[n - D_i; m], & D_i - N_{cp} \leq n \leq N-1.
\end{cases}
\end{align}
As in Case 2, we can show that the received signal on the $\ell$-th subcarrier during the $m$-th OFDM symbol is given by
\begin{align}
&Y[\ell; m]
= \sqrt{E_{i}}\frac{N-D_i+N_{cp}}{N} S_i[\ell;m] e^{- j 2 \pi \frac{\ell}{N} D_i } \notag \\ 
&+ \sqrt{E_{i}}\frac{D_i-N_{cp}}{N} S_i[\ell;m-1] e^{ - j 2 \pi \frac{\ell}{N} (D_i-N_{cp}) }   \notag \\
&+\sqrt{E_{i}} \frac{1}{N} \sum_{k \neq \ell} \left( \frac{1- e^{ j 2 \pi \frac{k-\ell}{N}(D_i-N_{cp}) }}{1- e^{ j 2 \pi \frac{k-\ell}{N} }} \right) \notag \\ 
&\times \left( - S_i [k; m] e^{- j 2 \pi \frac{k}{N} D_i } + S_i [k; m-1] e^{- j 2 \pi \frac{k}{N} (D_i-N_{cp}) } \right).
\end{align}
Thus, the total received power on the $\ell$-th subcarrier during the $m$-th OFDM symbol from transmitter $i$ is
\begin{align}
P_{i}[\ell; m] =& G_i[m]E_{i} \bigg( \frac{1}{N^2} \left((N-D_i + N_{cp})^2 + (D_i - N_{cp})^2\right) \notag \\
&
+ \frac{2}{N^2} \sum_{k \neq \ell} \frac{ \sin^2 \left( \pi \frac{D_i - N_{cp}}{N} (k-\ell) \right)}{ \sin^2 \left( \pi \frac{1}{N} (k-\ell) \right) } \bigg).
\label{eq:3}
\end{align}
If $S_i[\ell;m]$ is desired, the useful signal power is $\frac{(N-D_i + N_{cp})^2 }{N^2} G_i[m]E_{i}$; the remaining terms in (\ref{eq:3}) contribute to self-interference including both ICI and ISI. Otherwise, transmitter $i$'s signal appears as interference whose power is characterized by (\ref{eq:3}).

\subsection{From Link-Level to System-Level Studies}

In this subsection, we discuss how to apply the previous link-level analysis on the impact of timing misalignment to OFDM transmission in system-level studies. In an OFDM system, a transmitter sends a block of coded bits on the used subcarriers.  The probability that the receiver can decode the block sent by transmitter $i$ depends on all the $\sinr$ values of the used subcarriers. Transmitter $i$'s $\sinr$ of subcarrier $\ell$ is given by
\begin{align}
\sinr_i [\ell]  = \frac{ g(D_i)  G_i E_{i} }{ P_i[\ell] - g(D_i)  G_i E_{i}  + \sum_{j\neq i}  P_{j}[\ell]  + N_0 } ,
\label{eq:5}
\end{align}
where we have dropped the OFDM symbol index $m$, $N_0$ denotes the noise power, and
\begin{align}
g(d) = \left\{ 
 \begin{array}{l l}
    0 & \quad -(N+N_{cp}) \leq d< -N; \\
    \frac{(N+d)^2 }{N^2} & \quad -N \leq d < 0; \\
    1 & \quad 0\leq d < N_{cp}; \\
    \frac{(N + N_{cp} - d)^2 }{N^2} & \quad N_{cp} \leq d <  N+N_{cp}.
\end{array} \right.
\label{eq:4}
\end{align}

In a system-level study, the subcarrier $\sinr$ values are usually mapped to a unique $\sinr$, based on which the decision on whether the block is decodable is made. For example, the exponential effective SINR mapping (EESM) is a popular mapping method \cite{tuomaala2005effective}. In an asynchronous network with timing misalignment, the calculation of $\sinr_i [\ell]$ can be difficult because the detailed modeling of timing errors in a system-level study can be cumbersome. Further, the received power $P_{i}[\ell]$ depends on timing misalignment in a delicate way (c.f. (\ref{eq:2}) and (\ref{eq:3})), which makes the analytical evaluation of system-level performance even more challenging.

To solve the above mentioned difficulties, we propose a simple first-order model, which can be conveniently used in system-level studies.

\textbf{System-Level Abstraction.}
\textit{In a system-level study of the asynchronous network with timing misalignment, the subcarrier $\sinr_i [\ell]$ may be approximately calculated as follows.
\begin{enumerate}
\item Model and calculate the timing misalignment $D_i$ between transmitter $i$ and the typical receiver.
\item Calculate the useful signal power as $g(D_i)  G_i E_{i}$, where $g(d)$ is defined in (\ref{eq:4}).
\item Approximate the total received signal power from transmitter $j$ as $P_j[\ell]  = G_j E_j, j = 1,2,...$.
\item Calculate $\sinr_i [\ell]$ according to (\ref{eq:5}).
\end{enumerate}
}

The proposed system-level abstraction has two main advantages: (1) when evaluating $\sinr_i [\ell]$ it only needs to consider the timing misalignment of the receiver with respect to transmitters $i$; and (2) compared to the original complicated expressions (c.f. (\ref{eq:2}) and (\ref{eq:3})), the total received signal power from transmitter $j$ is simply approximated as $P_j[\ell]  = G_j E_j$. These two facts greatly simplify system-level studies.

The validness of the proposed system-level abstraction hinges on the condition that the total received signal power from transmitter $j$ can be well approximately as $P_j[\ell]  = G_j E_j$, regardless of the timing misalignment $D_j$. As shown in a numerical example in Fig. \ref{fig:1}, this approximation is quite accurate: the received powers are almost uniform on the used subcarriers except a few edge subcarriers under various timing misalignment cases. Fig. \ref{fig:2} further shows how the timing misalignment in OFDM transmission affects the power of useful signal as well as the power of self-interference. For example, the received SNR of the central subcarrier would be limited to less than $20$ dB when the receiving window is later than the actual timing of the received signal by $6$ samples (mainly due to the self-interference).

\begin{figure}
\centering
\includegraphics[width=8.5cm]{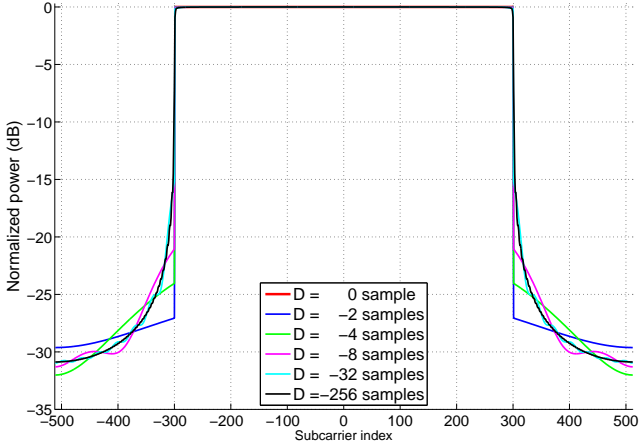}
\caption{Received power of an OFDM signal with timing misalignment. $N=1024; N_{cp}=72$; the used subcarriers are $\{-299,...,0,...300\}$. }
\label{fig:1}
\end{figure}

\begin{figure}
\centering
\includegraphics[width=8.5cm]{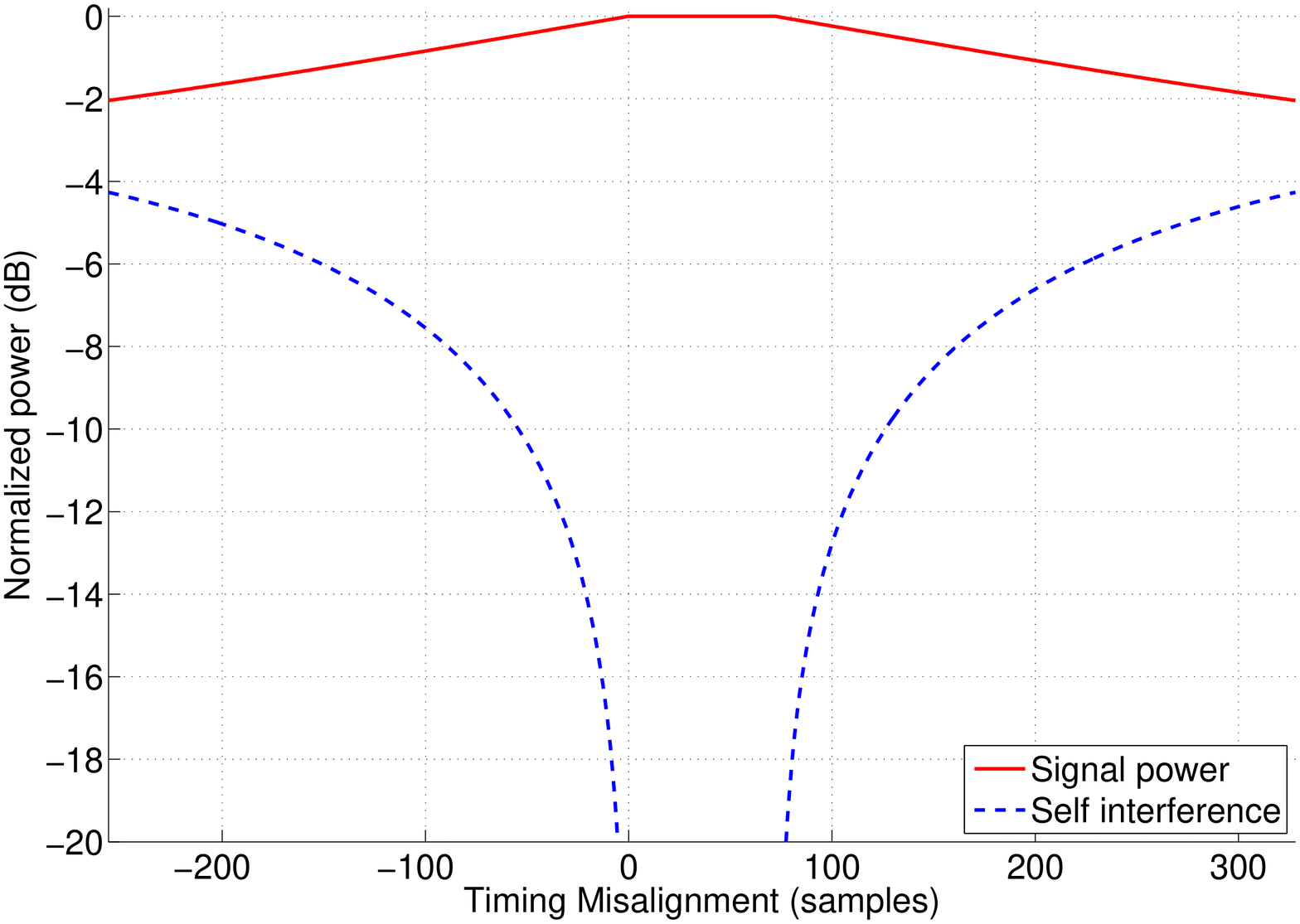}
\caption{Signal and self-interference powers of an OFDM signal received on the central subcarrier with timing misalignment. $N=1024; N_{cp}=72$; the used subcarriers are $\{-300,...,0,...299\}$. }
\label{fig:2}
\end{figure}

\section{On the Decodable Transmitters of a Typical Receiver}
\label{sec:on}

In this section, we apply the proposed system-level abstraction to study several important statistics about the transmitters whose packets can be decoded by the typical receiver in the asynchronous network. Such statistics include the average number of decodable transmitters, the decoding probability of the nearest transmitter, the distribution of the number of decodable transmitters, and system throughput.

To this end, we first notice that with the proposed system-level abstraction, the subcarrier $\sinr_i [\ell]$ now can be written as
\begin{align}
\sinr_i [\ell]  = \frac{ g(D_i)  G_i E_{i} }{ (1- g(D_i))  G_i E_{i}  + \sum_{j\neq i}  G_j E_j  + N_0 }.
\label{eq:6}
\end{align}
Noting that the right hand side of (\ref{eq:6}) is independent of $\ell$, we can simply use the subcarrier $\sinr_i [\ell]$ as the block $\sinr_i$, based on which the decision on whether a packet is decodable can be made. Therefore, in the sequel we drop the subcarrier index $\ell$ in (\ref{eq:6}) and treat it as a block $\sinr$. 

In the following system-level study we assume that (i) transmitters use constant transmit power $E$, (ii) the timing mismatches $\{D_i\}$ are i.i.d. with cumulative distribution function (CDF) $F_D(\cdot)$, and (iii) the channel gain $G_i$ is modeled as
\begin{align}
G_i = \|X_i\|^{-\alpha}  F_{i},
\end{align}
where $\alpha > 2$ is the pathloss exponent, and $F_{i}$ denotes the fading of the link from transmitter $i$ to the typical receiver. For simplicity, we consider independent Rayleigh fading in this paper, i.e., $F_{i} \sim \textrm{Exp} (1)$; more general fading and/or the effect of shadowing may be treated by further applying Displacement theorem for the PPP \cite{blaszczyszyn2013using}, which is not the focus of this paper. With these assumptions, the $\sinr_i$ now can be written as
\begin{align}
\sinr_i   = \frac{ g(D_i)  \|X_i\|^{-\alpha}  F_{i}  }{ (1- g(D_i))  \|X_i\|^{-\alpha}  F_{i}  + \sum_{j\neq i}  \|X_j\|^{-\alpha}  F_{j}  + N_0 /E}.
\end{align}

We let $\mathcal{E}_i$ be the event that a packet from transmitter $i$ is decodable. Then the event $\mathcal{E}_i$ occurs if and only if the received $\sinr_i$ is above some detection threshold $T$, which is a function of the used modulation and coding scheme. Mathematically, the number $\Upsilon$ of decodable transmitters is given by
\begin{align}
\Upsilon = \sum_{i} \mathbb{I} ( \mathcal{E}_i ) = \sum_{i} \mathbb{I} ( \sinr_i \geq T ) ,
\end{align}
where $\mathbb{I} ( \mathcal{E} )$ is an indicator function which equals $1$ if the
event $\mathcal{E}$ is true and $0$ otherwise. Clearly, $\Upsilon$ is a random variable and will be the central object studied in the sequel.

\subsection{Mean Number of Decodable Transmitters}
\label{subsec:mean}

We first consider the average number of decodable transmitters $\mathbb{E}[\Upsilon]$.

\begin{pro}
The mean number of decodable transmitters is given by
\begin{align}
\mathbb{E} [  \Upsilon ] = \pi \lambda & \int_{\mathbb{D}}  \int_0^\infty \mathbb{I} \left( g(\tau) > \frac{T}{1+T} \right)   e^{ - h(\tau, T) \snr^{-1} v^{\frac{\alpha}{2}} } \notag \\
&\times
  e^{  -  \lambda \pi \textrm{sinc}^{-1}\left( \frac{2}{\alpha} \right)  \left ( h(\tau,T) \right)^{\frac{2}{\alpha}} v  }    \dint v   F_D(\dint \tau) ,
\label{eq:7}
\end{align}
where $h(\tau,T) = \frac{T }{ (1+T) g(\tau) - T }$, $\snr = E/N_0$, and $\textrm{sinc} (x) = \frac{\sin (\pi x) }{\pi x}$. 
\label{pro:1}
\end{pro}
\begin{proof}
See Appendix \ref{proof:pro:1}.
\end{proof}

To gain some insights from Prop. \ref{pro:1}, we next focus on the special case that the network is interference-limited, i.e., $N_0\to 0$. 
\begin{cor}
In the interference-limited case with $N_0\to 0$, (\ref{eq:7}) reduces to a simpler form:
\begin{align}
\mathbb{E} [  \Upsilon ] 
&= \mathbb{E}_D \left[  \frac{\mathbb{I} \left( g(D) > \frac{T}{1+T} \right)  \textrm{sinc}\left( \frac{2}{\alpha} \right) }{ (h(D,T))^{\frac{2}{\alpha}}  }   \right] ,
\end{align}
which can be upper bounded as
\begin{align}
\mathbb{E} [  \Upsilon ]  \leq \frac{\textrm{sinc}\left( \frac{2}{\alpha} \right) }{ T^{\frac{2}{\alpha}}  } .
\label{eq:17}
\end{align}
\end{cor}

The upper bound (\ref{eq:17}) follows because by definition $g(\tau) \leq 1$ (c.f. (\ref{eq:4})) and thus $h(\tau,T) \geq T$ for all $\tau \in \mathbb{D}$ satisfying $g(\tau) > T/(1+T)$. The upper bound is attained when timing misalignment $D$ is restricted within the range of cyclic prefix. This simple upper bound only depends on two network parameters: $\alpha$ and $T$. In particular, the upper bound decreases as the detection threshold $T$ increases, agreeing with intuition: the mean number of decodable transmitters decreases when the modulation and coding rate are chosen such that $T$ is higher. 

\begin{figure}
\centering
\includegraphics[width=8.5cm]{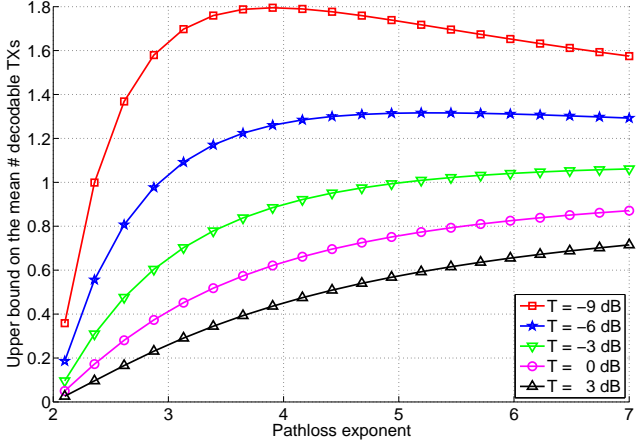}
\caption{The upper bound on the mean number of decodable transmitters (c.f. (\ref{eq:17})) versus pathloss exponent.}
\label{fig:3}
\end{figure}

The dependency of the upper bound on the pathloss exponent $\alpha$ is more complicated and is illustrated in Fig. \ref{fig:3}. 
Note that $\textrm{sinc}\left( \frac{2}{\alpha} \right)$ is increasing with $\alpha \in (2, \infty)$. In contrast,  when $0<T < 1$,  $T^{\frac{2}{\alpha}}$ is increasing with $\alpha \in (2, \infty)$, but when $T\geq 1$,  $T^{\frac{2}{\alpha}}$ is decreasing with $\alpha \in (2, \infty)$. Therefore, when $T \geq 1$, the upper bound increases with $\alpha \in (2, \infty)$. The intuition is that in order to decode packets from more transmitters in the median-to-high modulation and coding rate regime, it is important to reduce the interference power in the interference-limited scenario and thus high pathloss exponent is favorable.  When $0<T < 1$, it is possible that the upper bound first increases and then decreases as the pathloss exponent increases. This is because in the low modulation and coding rate regime, it is also important to preserve the useful signal power while reducing the interference power. In particular, for very low $T$, as $\alpha$ increases beyond some point, the loss of the useful signal power will outweigh the gain of interference reduction and thus the mean number of decodable transmitters will eventually decrease. Another interesting observation from Fig. \ref{fig:3} is that the mean number of decodable transmitters is very small: it is less than $2$ even when $T$ is as low as $-9$ dB. We will explore this fact more in later sections.

Though the above discussion is carried out in the interference-limited case, the overall insights still hold when noise is taken into account. For example, Fig. \ref{fig:33} considers noise (whose power is given in Table \ref{tab:sys:para}) and shows the performance under two two transmitter densities. The dense case with $\lambda = 1/20^2$ m$^{-2}$ is interference-limited; in this case, we can see that the upper bound shown in Fig. \ref{fig:3} is quite close to the true values shown in Fig. \ref{fig:33}. In the sparse case with $\lambda = 1/400^2$ m$^{-2}$ where the noise has a more pronounced effect,  Fig. \ref{fig:33} shows that a moderate pathloss exponent (around $3.3$) is preferred as it strikes a balance between interference reduction and preserving the useful signal power.

\begin{figure}
\centering
\includegraphics[width=8.5cm]{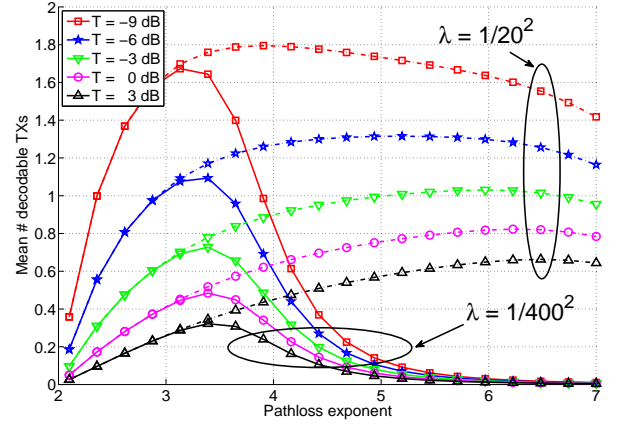}
\caption{Mean number of decodable transmitters versus pathloss exponent in synchronized networks.}
\label{fig:33}
\end{figure}

Next let us turn to the impact of timing misalignment. As expected and shown in (\ref{eq:17}), there is a loss in the mean number of decodable transmitters due to the timing  misalignment. However, if the timing misalignment is restricted within the range of cyclic prefix, i.e., $D \in [0, N_{cp})$, then $g(D) \equiv 1$ and thus the upper bound in (\ref{eq:17}) is attained. In this case, there is no loss due to the timing misalignment. Otherwise, the loss exists and depends on the distribution of the timing misalignment. Note that the integrand in (\ref{eq:7}) is zero if $g(\tau) \leq T/(1+T)$. The physical interpretation is that when $g(\tau) \leq T/(1+T)$, the self interference caused by timing misalignment is already large enough to cause the decoding failure.

To obtain a more concrete understanding of the impact of timing misalignment, we show some numerical results in the sequel. As a null hypothesis, we assume that the distribution of the timing misalignment is Gaussian with mean $0$ and standard deviation $\sigma$ but is truncated within the range $[ -(N+N_{cp}), N+N_{cp} )$. The specific parameters used in plotting numerical or simulation results in this paper are summarized in Table \ref{tab:sys:para} unless otherwise specified. Note that, with the OFDM sampling period normalized to $1$, $N$ denotes the duration of the data part of an OFDM symbol. Accordingly, we normalize timing error deviation $\sigma$ and measure it in terms of $N$, as indicated in Table \ref{tab:sys:para}. 

Fig. \ref{fig:4} shows the mean number of decodable transmitters versus the detection threshold. From Fig. \ref{fig:4}, we can see that asynchronous transmissions have a remarkable effect on the performance; for example, when aiming at decoding one transmitter on average and $\lambda = 1/20^2$ m$^{-2}$, the loss in the supported detection threshold is about $2$ dB (resp. $4$ dB) with $\sigma = 0.2N$ (resp. $\sigma = 0.4N$). Similarly, with the detection threshold $T=-4$ dB, the loss in the mean number of decodable transmitters is $21\%$ (resp. $44\%$) when $\sigma = 0.2N$ (resp. $\sigma = 0.4N$). Fig. \ref{fig:4} also shows that the relative loss in the mean number of decodable transmitters due to asynchronous transmissions increases as the detection threshold increases, implying that asynchronous transmissions have a more significant impact on high-rate communication. Similar observations hold when $\lambda = 1/400^2$ m$^{-2}$. Note that the simulation results clearly match the analysis in Fig. \ref{fig:4}; this provides a sanity check for the derived analytical results.

\begin{table}
\centering
\begin{tabular}{|l||r|} \hline
Tx density $\lambda$   & ${1}/{400^2}$ m$^{-2}$  \\ \hline 
PL exponent  $\alpha$ & $3.8$ \\ \hline
Tx power   & $23$ dBm  \\ \hline
Channel bandwidth    & $10$ MHz \\ \hline
Noise PSD    & $-174$ dBm \\ \hline
Rx noise figure    & $9$ dB \\ \hline
Detection threshold $T$ & $-12$ dB \\ \hline
($N, N_{cp}$)   & (1024, 72) \\ \hline
Timing error deviation $\sigma$  & $0.2N$ \\ \hline
\end{tabular}
\caption{Simulation/Numerical Parameters}
\label{tab:sys:para}
\end{table}

\begin{figure}
\centering
\includegraphics[width=8.5cm]{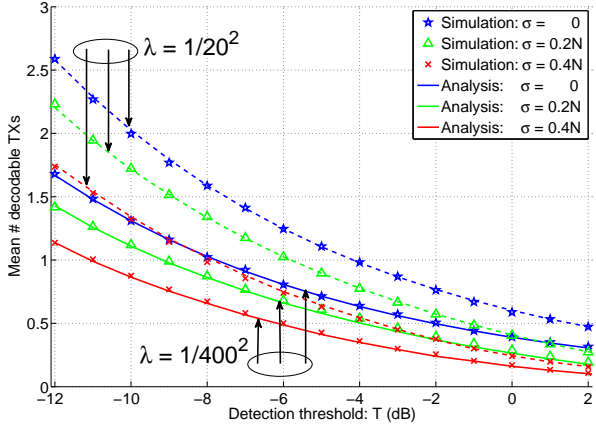}
\caption{Mean number of decodable transmitters versus detection threshold.}
\label{fig:4}
\end{figure}

%

\subsection{An Upper Bound on the Distribution of the Number of Decodable Transmitters}

In the previous subsection, we studied the first order statistic of the number $\Upsilon$ of decodable transmitters. In this subsection, we take a broader view and study the distribution of the number $\Upsilon$ of decodable transmitters. Though an exact characterization is possible, the resulting expressions involve very high dimensional integrals even in the case of perfectly synchronized networks \cite{keeler2013sinr}. Instead, we provide a simple upper bound on the distribution of $\Upsilon$ in the following proposition.
\begin{pro}
The number $\Upsilon$ of decodable transmitters is (first order) stochastically dominated by a truncated Poisson random variable $\Upsilon^{(\textrm{u})}$, i.e., $\mathbb{P}( \Upsilon^{(\textrm{u})} \geq n ) \geq \mathbb{P}( \Upsilon \geq n ),  n=0,1,...$. The distribution of $\Upsilon^{(\textrm{u})}$ is given as follows: $\mathbb{P}( \Upsilon^{(\textrm{u})} = n ) =\frac{1}{C} \frac{\tilde{\lambda}^n}{n!} , n=0,..., \lfloor \frac{1+T}{T} \rfloor$, where 
\begin{align}
\tilde{\lambda} = \pi \lambda \int_0^{\infty} \mathbb{E}_D \left[ \mathbb{I} \left( g(D) > \frac{T}{1+T} \right) e^{ - \frac{T v^{{\alpha}/{2}}}{  g(D) \snr  } }  \right]    \dint v ,
\end{align}
and $C$ is a normalization constant such that $\sum_{n=0}^{\lfloor \frac{1+T}{T} \rfloor} \mathbb{P}( \Upsilon^{(\textrm{u})} = n )=1$.
\label{pro:3}
\end{pro}
\begin{proof}
See Appendix \ref{proof:pro:3}.
\end{proof}

To gain some insights from Prop. \ref{pro:3}, we next focus on the special case with $T>1$. Then Prop. \ref{pro:1} implies that $\Upsilon^{(\textrm{u})}$ is a Bernoulli random variable: it equals $1$ with probability ${\tilde{\lambda}}/{(1+\tilde{\lambda})}$ and $0$ otherwise. The mean of  $\Upsilon^{(\textrm{u})}$ is ${\tilde{\lambda}}/{(1+\tilde{\lambda})}$. If the network is very sparse such that $\lambda \sim o(1)$, then ${\tilde{\lambda}}/{(1+\tilde{\lambda})} \sim \tilde{\lambda} = \Theta(\lambda)$. When the transmit power is fixed, the performance of sparse networks is noise-limited. This indicates that in the noise-limited case the probability that the receiver can decode a packet from some transmitter is $O(\lambda)$. So is the mean number of decodable transmitters. In the next subsection, we will show that the probability is $\Omega(\lambda)$ as $\lambda \to 0$, and thus the probability actually scales as $\Theta (\lambda)$.

If the network is very dense, i.e., $\lambda \to \infty$, then ${\tilde{\lambda}}/{(1+\tilde{\lambda})} \sim 1$. Clearly, the performance of dense networks is interference-limited. As a result, one might think that in the interference-limited case the probability that the receiver can decode a packet from some transmitter is close to $1$. The fallacy of the above argument is that ${\tilde{\lambda}}/{(1+\tilde{\lambda})}$ is an upper bound and may not be tight as $\lambda \to \infty$. In fact, the right intuition should be that the received SINR from any transmitter in the interference-limited case would not be large and thus the probability that no transmitter can be decoded can be relatively high if the detection threshold $T$ is large. The last intuition can be further confirmed by examining Fig. \ref{fig:3}. For example, Fig. \ref{fig:3} shows that the mean number of decodable transmitters is less than $0.5$ at $\alpha=4$ and $T=3$ dB, implying that the probability that no transmitter can be decoded is greater than $0.5$.

Note that the parameter $\tilde{\lambda}$ may take more explicit form in some special cases. For example, when $\alpha = 4$, 
\begin{align}
\tilde{\lambda}  = \frac{\pi^{\frac{3}{2}} \lambda }{2} \sqrt{\frac{\snr}{T}}  \mathbb{E}_D \left[   \mathbb{I} \left( g(D) > \frac{T}{1+T} \right) \sqrt{g(D)} \right] .
\end{align}
Therefore, if $T>1$ and $\alpha=4$,  the probability that the receiver can decode a packet from some transmitter is proportional to the square root of $\snr$ in the noise-limited case, agreeing with intuition: the radio link length is proportional to $\snr^{1/4}$ when $\alpha=4$ and thus the decoding probability should be proportional to $\snr^{1/2}$ in $\mathbb{R}^2$. Similar intuition may be used to explain why the probability is inversely proportional to the square root of the detection threshold $T$.

Figs. \ref{fig:6} and \ref{fig:7} compare the analytical upper bound on the distribution of the number $\Upsilon$ of decodable transmitters to the corresponding true distribution obtained from simulation under two different transmitter densities. It can be seen that the analytical upper bound is more accurate when the network is sparser (i.e. less interference-limited).

\begin{figure}
\centering
\includegraphics[width=8.5cm]{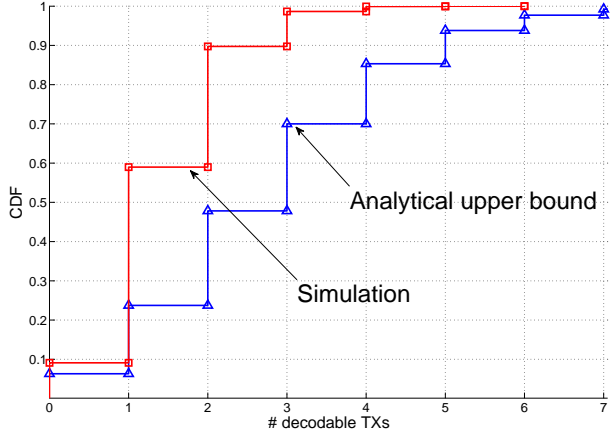}
\caption{Analytical upper bound vs. simulation on the distribution of the number of decodable transmitters: $\lambda = 1/400^2$ m$^{-2}$.}
\label{fig:6}
\end{figure} 

\begin{figure}
\centering
\includegraphics[width=8.5cm]{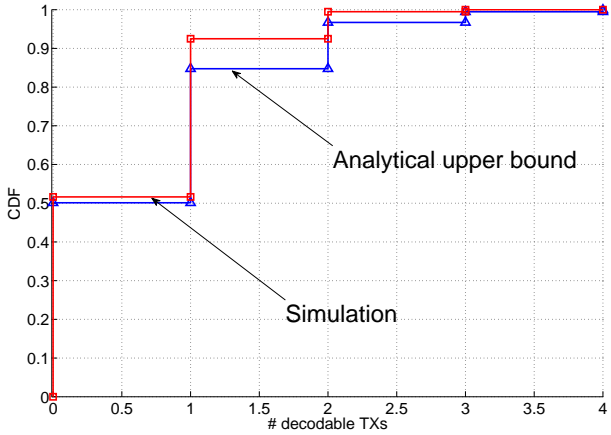}
\caption{Analytical upper bound vs. simulation on the distribution of the number of decodable transmitters: $\lambda = 1/800^2$ m$^{-2}$.}
\label{fig:7}
\end{figure} 

\subsection{On Decoding the Nearest Transmitter}

According to Prop. \ref{pro:3}, the receiver can decode the packet from at most one transmitter if $T>1$. The decodable transmitter is typically the nearest one, though fading and timing misalignment may affect the result. Further, the probability of decoding the nearest transmitter indicates the coverage performance of cellular networks where the positions of BSs are modeled by a PPP \cite{andrews2011tractable}. Therefore, it is of particularly interest to study the probability that the receiver can decode a packet sent by the nearest transmitter. We answer this question in the following Proposition \ref{pro:2}.
\begin{pro}
The probability that the receiver can decode a packet sent by the nearest transmitter $X_0$ is given by
\begin{align}
&\mathbb{P} ( \textrm{SINR}_0 \geq T  )  
=   \pi \lambda \int_{\mathbb{D}} \int_0^\infty \mathbb{I} \left( g(\tau) > \frac{T}{1+T} \right)   \notag \\
&\times e^{ -h(\tau, T) \snr^{-1}  v^{\frac{\alpha}{2}} } e^{ -\pi \lambda ( 1 + \rho (h(\tau, T) ,\alpha ) ) v } \dint v F_D(\dint \tau) ,
\label{eq:11}
\end{align}
where 
$
\rho (x, \alpha) = x^{\frac{2}{\alpha}} \int_{x^{-\frac{2}{\alpha}}}^\infty \frac{1}{1 + v^{\frac{\alpha}{2}} } \dint v 
$, and $h(\tau, T)$ is defined in Prop. \ref{pro:1}.
\label{pro:2}
\end{pro}
\begin{proof}
See Appendix \ref{proof:pro:2}.
\end{proof}

From Prop. \ref{pro:2}, it is easy to see the probability that the receiver can decode a packet sent by the nearest transmitter is $\Theta(\lambda)$ as $\lambda \to 0$. Thus, the probability that the receiver can decode a packet sent by at least one transmitter is $\Omega(\lambda)$ as $\lambda \to 0$. The last fact has been used in the previous section when stating that with $T>1$ the probability that the receiver can decode a packet sent by some transmitter scales as $\Theta (\lambda)$.

When the network is interference-limited, i.e., $N_0\to 0$,  (\ref{eq:11}) reduces to
\begin{align}
\mathbb{P} ( \textrm{SINR}_0 \geq T  )  &= \mathbb{E}_D \left[  \frac{ \mathbb{I} \left( g(D) > \frac{T}{1+T} \right) }{1 + \rho (h(D, T) ,\alpha )  }  \right]  \notag \\
&\leq  \frac{1}{ 1 + \rho(T, \alpha) } ,
\end{align}
where we have used the fact that $h(\tau,T) \geq T$, for all $\tau \in \mathbb{D}$ satisfying $g(\tau) > T/(1+T)$, in the last inequality. The above upper bound is attained when  $D \equiv 0$, i.e., the network is perfectly synchronized, which has been studied in \cite{andrews2011tractable}. In fact, as long as the timing misalignment $D$ is restricted within the range of cyclic prefix, the upper bound can be attained. As in the case of the mean number of decodable transmitters, there is a loss in the probability of decoding the nearest transmitter due to the timing  misalignment, and the loss depends on the distribution of the timing misalignment.

Fig. \ref{fig:8} shows the decoding probability of the nearest transmitter versus the detection threshold. From Fig. \ref{fig:8}, we can see that, when aiming at decoding probability $0.5$ and $\lambda = 1/20^2$ m$^{-2}$, the loss in the supported detection threshold is about $3$ dB (resp. $6$ dB) with $\sigma = 0.2N$ (resp. $\sigma = 0.4N$). Fig. \ref{fig:8} also shows that the impact of asynchronous transmissions becomes more significant as the detection threshold increases. Similar observations hold when $\lambda = 1/400^2$ m$^{-2}$. 

\begin{figure}
\centering
\includegraphics[width=8.5cm]{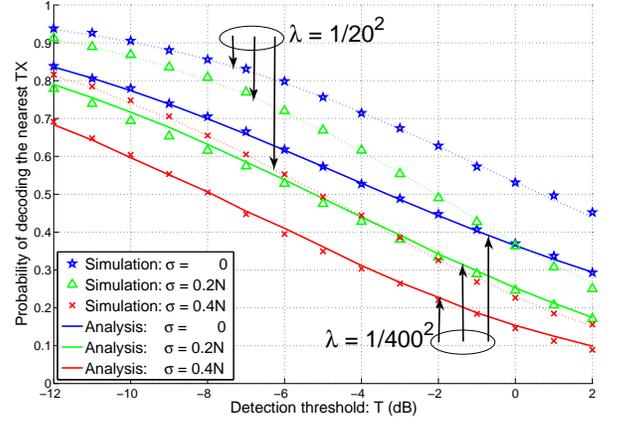}
\caption{Impact of timing  misalignment on the decoding probability of the nearest transmitter.}
\label{fig:8}
\end{figure}

%

\subsection{Optimizing System Throughput}

The average number of decodable transmitters characterized in Prop. \ref{pro:1} is monotonically increasing as the detection threshold $T$ decreases. However, reducing the detection threshold $T$ implies that we adopt lower modulation order and/or coding rate. This may be undesirable from a throughput point of view. In order to take into account this tradeoff, we define \textit{system throughput} $\xi$ as the mean of the sum rate of all the transmitters to the typical receiver. Mathematically,
\begin{align}
\xi = \mathbb{E} \left[ \sum_i \mathbb{I} ( \sinr_i \geq T ) \log (1+T) \right].
\end{align}
With this definition, the following result follows immediately.
\begin{cor}
The system throughput equals
$\xi = \log(1+T)  \mathbb{E}[\Upsilon ]$ with $\mathbb{E}[\Upsilon ]$ given in Prop. \ref{pro:1}.
\end{cor}

Now we may optimize the detection threshold $T$ by maximizing the system throughput $\xi$. This optimization is of single variable and thus can be solved efficiently. To gain some intuition, we show the system throughput as
a function of $T$ in Fig. \ref{fig:10}. From Fig. \ref{fig:10}, we can see that the optimal detection thresholds are respectively $5$ dB, $-1$ dB and $-3$ dB when $\sigma = 0$, $0.2N$ and $0.4N$. This implies that we have to be more conservative in setting the detection threshold when the networks are asynchronous (vs. synchronized networks). Another interesting observation from Fig. \ref{fig:10} is that the optimal detection thresholds are nearly unaffected by the transmitter density.

\begin{figure}
\centering
\includegraphics[width=8.5cm]{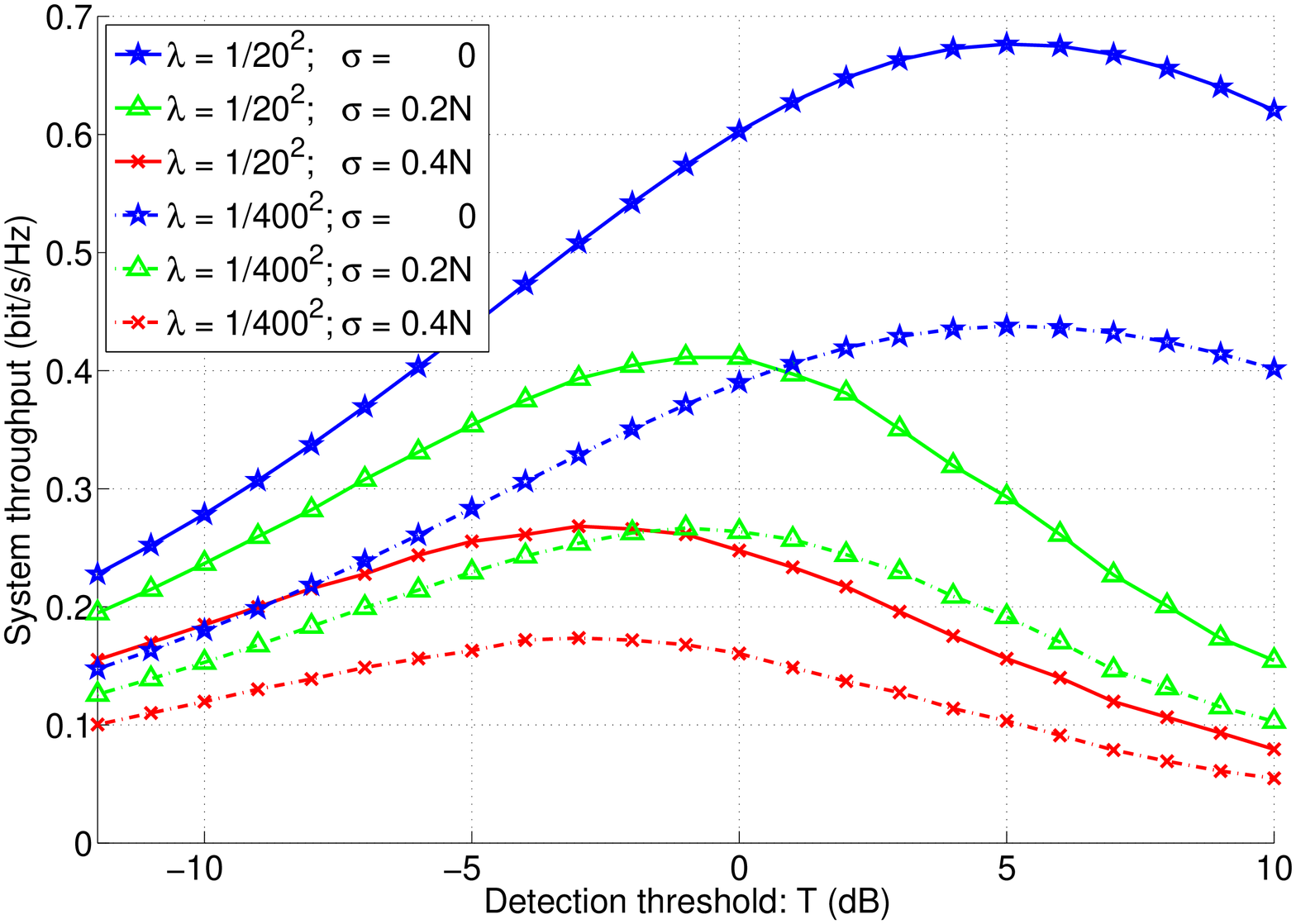}
\caption{System throughput versus detection threshold.}
\label{fig:10}
\end{figure}

\section{Solutions to Mitigating the Loss of Asynchronous Transmissions}
\label{sec:solutions}

In the previous section, we have seen that asynchronous transmissions may have a remarkable effect on the system-level performance. In this section we discuss four possible solutions, which differ in complexity and may be applicable in different scenarios, to mitigate the loss due to asynchronous transmissions.

\subsubsection{Extended cyclic prefix}

If the timing mismatches are concentrated in the range $[0, N_x)$ where $N_x > N_{cp}$, we can solve the timing misalignment problem by simply extending the length of the cyclic prefix beyond $N_x$. However, using cyclic prefix of extended length comes at the cost of more power and time spent in sending the cyclic prefix instead of being used to communicate data. This is a tradeoff, the characterization of which is beyond the scope of this paper.
The general principle is that this approach is applicable to the scenarios where $N_x$ is not too large.

\subsubsection{Advanced receiver timing}

If the timing mismatches are concentrated in the range $[-N_x, N_y)$ where $N_x, N_y >0$ and $N_x + N_y \leq N_{cp}$, we can solve the timing misalignment problem by simply advancing the receiver timing by $N_x$. Then the timing mismatches will be concentrated in the range $[0, N_x + N_y)$. As $N_x + N_y \leq N_{cp}$, there will be no loss due to the timing misalignment after shifting the receiver's timing earlier.  This approach is very simple but is only applicable to the scenarios where $N_x+N_y \leq N_{cp}$, and it also requires knowledge of $N_x$.

\subsubsection{Dynamic receiver timing positioning}

The receiver may estimate the timings used by each transmitter through either pilot-based or non pilot-based synchronization methods. Once a transmitter's timing is obtained, the receiver can adaptively adjust its receiving window to decode the transmitter's packet. Compared to the previous two approaches, dynamic receiver timing positioning is applicable to many more scenarios but at the cost of higher complexity. In particular, as the transmitters have i.i.d. timing mismatches, the typical receiver needs to estimate every transmitter's timing and accordingly positions its receiving window to decode a transmitter's packet.  

\subsubsection{Semi-static receiver timing positioning with multiple timing hypotheses}

Instead of estimating each transmitter's timing, the receiver may simply adopt multiple timing hypotheses: $-n_1 \Delta,...,0,...,n_2\Delta$, where $n\Delta$ denotes the timing difference between the hypothesis $n$ and the receiver's timing. For every timing hypothesis, the receiver accordingly adjusts its receiving window and performs decoding; the packets from the transmitters whose timings happen to be around the current timing hypothesis may be decoded. This semi-static receiver timing positioning approach reduces the complexity of dynamic receiver timing positioning but still requires the receiver to use multiple timing windows. Further, a careful choice of $n_1,n_2$ and $\Delta$ is important for the design. In general, the more the used timing hypotheses, the smaller the loss due to timing misalignment but the higher the complexity.

The above proposed solutions may be combined depending on the application scenarios. For example, advanced receiver timing may be jointly used with extended cyclic prefix to make the condition $N_x+N_y \leq N_{cp}$ hold. In practice, the design decision on which solution should be used or how they should be combined is best made based on the specific scenario under consideration. Note that if our target is not to decode as many transmitters as possible but for example is to decode the nearest transmitter, synchronizing directly with the nearest transmitter is of reasonable complexity and recovers the loss.

Let us consider the solution of semi-static receiver timing positioning with multiple timing hypotheses since it can be applied to many scenarios while having reasonable complexity. We take 
the mean number of decodable transmitters as the metric to evaluate its effectiveness. The following corollary immediately follows from Prop. \ref{pro:1}.
\begin{cor}
Denote by $\mathcal{H} = \{ -n_1 \Delta,...,0,...,n_2\Delta \}$ the set of timing hypotheses. The mean number of decodable transmitters is given by (\ref{eq:7}) but with $g(x)$ substituted by $\tilde{g}(x) \triangleq \max_{\tau \in \mathcal{H}} g(x - \tau)$.
\label{cor:2}
\end{cor}
The rationale of Corollary \ref{cor:2} is straightforward: a transmitter is decodable as long as it is decodable under any of the used timing hypotheses. Fig. \ref{fig:11} shows the effectiveness of using multiple timing hypotheses. As expected, the more the used timing hypotheses, the more the mean number of decodable transmitters. Also, we can see from Fig. \ref{fig:11} that in this numerical example using $3$ timing hypotheses helps recover the majority of the loss.

\begin{figure}
\centering
\includegraphics[width=8.5cm]{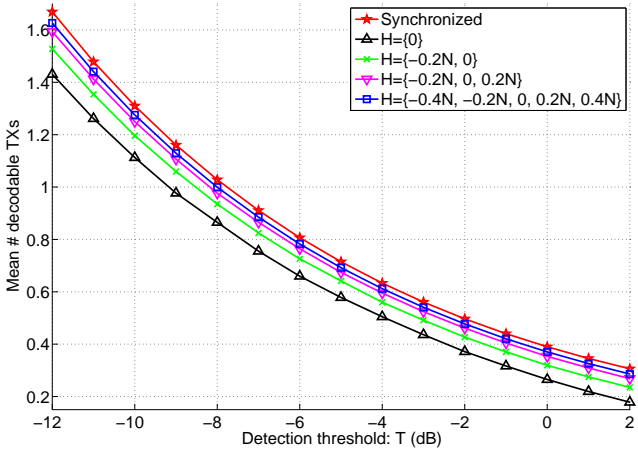}
\caption{Using semi-static receiver timing positioning with multiple timing hypotheses to mitigate the loss of asynchronous transmissions.}
\label{fig:11}
\end{figure}

\section{Conclusions}
\label{sec:conclusion}

In view of the lack of network-wide synchronization in many wireless networks, this paper has presented a baseline SINR model for asynchronous OFDM networks, which can be conveniently used in system-level studies. The model is then applied to characterize several important statistics in asynchronous PPP networks, including the number of decodable transmitters, the decoding probability of the nearest transmitter, and system throughput. The derived results complement existing analysis of synchronized networks using stochastic geometry. Further, this paper has compared and discussed four possible solutions for mitigating the loss of asynchronous transmissions.

This work can be extended in a number of ways. An extension from the studied flat-fading channels to frequency-selective channels would be highly desirable, since OFDM's main application is in such channels. It would also be of interest to explore scenarios where the notions of timing of different transmitters are not i.i.d. For example, a cluster of transmitters may synchronize with a common anchor node or base station and thus their notions of timing may become correlated. Also, if propagation delays are modeled, farther transmitters may likely have larger timing offsets with respect to the reference receiver, leading to non i.i.d. timing mismatches. One may further consider extending this work to study non-PPP network models and compare their performance to that of PPP model studied in this paper.

\appendix

\subsection{Derivation of Equations (\ref{eq:0}) and (\ref{eq:1})}
\label{app:deri}
We first derive (\ref{eq:0}). By the definition of discrete-time Fourier transform,  
\begin{align}
&Y[\ell; m] = \sum_{n=0}^{N-1} y[n;m] e^{- j 2 \pi \frac{\ell}{N} n } \notag \\
&= \sum_{n=0}^{N-1} s_i [n -D_i - N - N_{cp} ; m+1] e^{- j 2 \pi \frac{\ell}{N} n } \notag \\
&= \sum_{n=0}^{N-1}  \frac{\sqrt{E_{i}}}{N} \sum_{k} S_i [k; m+1] e^{ j 2 \pi \frac{k}{N} ( n - D_i - N - N_{cp}  ) } e^{- j 2 \pi \frac{\ell}{N} n } \notag \\
&= \sqrt{E_{i}} \sum_{k} S_i [k; m+1] e^{ j 2 \pi \frac{k}{N} (-D_i - N_{cp}) } \frac{1}{N} \sum_{n=0}^{N-1} e^{ j 2 \pi \frac{k-\ell}{N} n } \notag \\
&= \sqrt{E_{i}} e^{ j 2 \pi \frac{\ell}{N} (-D_i - N_{cp}) } S_i [\ell; m+1] ,
\end{align}
where we have plugged (\ref{eq:12}) in the second equality and used (\ref{eq:13}) in the third equality, and the last equality follows from the fact that $\frac{1}{N} \sum_{n=0}^{N-1} e^{ j 2 \pi \frac{k-\ell}{N} n } = \delta[k-\ell]$.

We next derive (\ref{eq:1}). Using the definition of discrete-time Fourier transform,  (\ref{eq:14}) and (\ref{eq:13}) yields
\begin{align}
&Y[\ell; m] = \sum_{n=0}^{N-1} y[n;m] e^{- j 2 \pi \frac{\ell}{N} n } \notag \\
&= \sum_{n=0}^{N-1 + D_i} s[-D_i+n; m] e^{- j 2 \pi \frac{\ell}{N} n } \notag \\
&+ \sum_{n=N + D_i}^{N-1} s[n - (N + D_i) - N_{cp}; m+1] e^{- j 2 \pi \frac{\ell}{N} n } \notag \\
&= \sum_{n=0}^{N-1 + D_i}    \frac{\sqrt{E_{i}}}{N} \sum_{k} S_i [k; m] e^{ j 2 \pi \frac{k}{N} (-D_i+n) } e^{- j 2 \pi \frac{\ell}{N} n } + \notag \\
&\sum_{n=N+D_i}^{N-1}    \frac{\sqrt{E_{i}}}{N} \sum_{k} S_i [k; m+1] e^{ j 2 \pi \frac{k}{N} (n - (N+D_i) - N_{cp}) } e^{- j 2 \pi \frac{\ell}{N} n } .
\label{eq:15}
\end{align}
The first sum in (\ref{eq:15}) equals
\begin{align}
&\sqrt{E_{i}}\frac{N+D_i}{N} S_i[\ell;m] e^{- j 2 \pi \frac{\ell}{N} D_i } \notag \\ 
&+ \sqrt{E_{i}} \frac{1}{N} \sum_{k \neq \ell} S_i [k; m] e^{- j 2 \pi \frac{k}{N} D_i } \sum_{n=0}^{N-1+D_i} e^{ j 2 \pi \frac{k-\ell}{N} n } ,
\end{align}
and the second sum in (\ref{eq:15}) equals
\begin{align}
&-\sqrt{E_{i}}\frac{D_i}{N} S_i[\ell;m+1] e^{ j 2 \pi \frac{\ell}{N} (-D_i-N_{cp}) } 
+ \sqrt{E_{i}} \frac{1}{N}  \notag \\ 
& \times \sum_{k \neq \ell} S_i [k; m+1] e^{ j 2 \pi \frac{k}{N} (-D_i-N_{cp}) } \sum_{n=N+D_i}^{N-1} e^{ j 2 \pi \frac{k-\ell}{N} n }.
\end{align}
Combining the above two results, and plugging in the following two equations
\begin{align}
&\sum_{n=0}^{N-1+D_i} e^{ j 2 \pi \frac{k-\ell}{N} n } = \frac{1- e^{ j 2 \pi \frac{k-\ell}{N}(N+D_i) }}{1- e^{ j 2 \pi \frac{k-\ell}{N} }} \notag \\
&\sum_{n=N+D_i}^{N-1} e^{ j 2 \pi \frac{k-\ell}{N} n } = \frac{ e^{ j 2 \pi \frac{k-\ell}{N} (N+D_i) } (1- e^{ j 2 \pi \frac{\ell-k}{N} D_i } )}{1- e^{ j 2 \pi \frac{k-\ell}{N} }},
\end{align}
we obtain (\ref{eq:1}).

\subsection{Proof of Proposition \ref{pro:1}}
\label{proof:pro:1}

For notational simplicity, denote by $I_{\Phi} = \sum_{X_j \in \Phi}   \|X_j\|^{-\alpha}  F_{j}$. Then by definition,
\begin{align}
&\mathbb{E} [  \Upsilon ] = \notag \\
&\mathbb{E} \bigg[  \sum_{X_i \in \Phi} \mathbb{I} \bigg(  
 \frac{ g(D_i)   \|X_i\|^{-\alpha}  F_{i}  }{ (1- g(D_i))   \|X_i\|^{-\alpha}  F_{i}  + I_{\Phi - \delta_{X_i}}  + N_0 /E } \geq T \bigg) \bigg] \notag \\
&= \int_{\mathbb{R}^2} \int_{\mathbb{D}}  \int_{\mathbb{R}}  \mathbb{E} \left[  \mathbb{I} \left( \frac{ g(\tau)   \|x\|^{-\alpha}  f  }{ (1- g(\tau))   \|x\|^{-\alpha}  f  + I_{\Phi}  + N_0 /E } \geq T \right)   \right ] \notag \\
& \quad \quad \quad   \quad \quad F_F(\dint f) F_D(\dint \tau)  M( \dint x)  ,
\end{align}
where $M(\cdot)$ is the mean measure of the PPP $\Phi$, i.e., $M(A) = \mathbb{E}[\Phi(A)]$ for any measurable set $A \subset \mathbb{R}^2$, and we have used the reduced Campbell formula for the PPP \cite{baccelli2009stochastic} in the last equality. Noting that $F_i$'s are i.i.d. Rayleigh fading, $F_F(\dint f) = e^{-f} \dint f, f \geq 0$. For the homogeneous PPP $\Phi$, $M(\dint x) = \lambda \dint x$. Using these two facts and changing the integral with respect to $x \in \mathbb{R}^2$ into polar coordinates, we have
\begin{align}
&\mathbb{E} [  \Upsilon ] 
= 2\pi \lambda \int_0^\infty \int_{\mathbb{D}}     \int_0^\infty \mathbb{E} \bigg[ \mathbb{I}\bigg(  \notag \\
& \quad
\frac{ g(\tau)   r^{-\alpha}  f  }{ (1- g(\tau)) r^{-\alpha}  f  + I_{\Phi}  + N_0 /E } \geq T \bigg) \bigg]  e^{-f}\dint f   F_D(\dint \tau)  r\dint r \notag \\
&=  2\pi \lambda \int_0^\infty \int_{\mathbb{D}}     \int_0^\infty  \mathbb{E} \bigg[ \mathbb{I} \left( g(\tau) > \frac{T}{1+T} \right) \notag \\
& \quad \quad \quad \times
\mathbb{I}\left(  f \geq r^\alpha h(\tau, T)( I_{\Phi}  + N_0 /E  )  \right) \bigg]  e^{-f}\dint f   F_D(\dint \tau) r\dint r   \notag \\
&=  2\pi \lambda \int_0^\infty   \int_{\mathbb{D}} \mathbb{I} \left( g(\tau) > \frac{T}{1+T} \right) \notag \\
& \quad \quad \quad  \quad \quad \times
e^{ -r^\alpha h(\tau, T) N_0 /E  }   
 \mathbb{E} \left[   e^{ - r^\alpha h(\tau, T) I_{\Phi}  }  \right]    F_D(\dint \tau)  r\dint r \notag \\
&=  2\pi \lambda \int_{\mathbb{D}}  \int_0^\infty  \mathbb{I} \left( g(\tau) > \frac{T}{1+T} \right) e^{ -r^\alpha h(\tau, T) N_0 /E  } \notag \\
& \quad \quad \quad  \quad \quad \times
  e^{  -  \lambda \pi \textrm{sinc}^{-1}\left( \frac{2}{\alpha} \right)  (h(\tau,T))^{\frac{2}{\alpha}} r^2    }   r\dint r   F_D(\dint \tau) ,
\label{eq:8}
\end{align}
where we have used the shorthand function $h(\tau, T)$ in the second equality and applied in the last equality the Laplace transform of the interference generated by a Poisson field of interferers with Rayleigh fading \cite{haenggi2009interference}:
\begin{align}
\mathcal{L}_{I_{\Phi}} (s) \triangleq \mathbb{E} [ e^{-s I_{\Phi} } ] = \exp \left(  -  \frac{\lambda \pi s^{\frac{2}{\alpha}} }{\textrm{sinc} \left( \frac{2}{\alpha} \right) }  \right).
\end{align}
With a change of variables $r^2 \to v$ in (\ref{eq:8}), we obtain (\ref{eq:7}) and complete the proof. 

\subsection{Proof of Proposition \ref{pro:3}}
\label{proof:pro:3}

The set of transmitters in $\Phi$ whose packets can be decoded can be upper bounded as 
\begin{align}
\tilde{\Phi} &= \sum_{X_i \in \Phi} \delta_{X_i}  \mathbb{I} ( \sinr_i \geq T )  \notag \\
&\leq \sum_{X_i \in \Phi} \delta_{X_i}  \mathbb{I} \left( \frac{ g(D_i)  F_i \|X_i\|^{-\alpha}   }{  N_0 /E} \geq T \right) \mathbb{I} \left( g(D_i) > \frac{T}{1+T} \right) \notag \\
&\triangleq \tilde{\Phi}^{(\textrm{u})} .
\end{align}
Note that given $\Phi$ the Bernoulli random variables $\mathbb{I} \left( \frac{ g(D_i)  F_{i} \|X_i\|^{-\alpha}   }{  N_0 /E} \geq T \right) \mathbb{I} \left( g(D_i) > \frac{T}{1+T} \right), i=1,2,...,$ are independent. It follows that $\tilde{\Phi}^{(\textrm{u})}$ is an independent thinning of $\Phi$ with thinning probability
\begin{align}
&p (x) = \mathbb{E} \left[ \mathbb{I} \left( \frac{ g(D_i)  F_{i} \|X_i\|^{-\alpha}   }{  N_0 /E} \geq T \right) \mathbb{I} \left( g(D_i) > \frac{T}{1+T} \right) \right]  \notag \\
&= \mathbb{E}_D \left[ \mathbb{I} \left( g(D) > \frac{T}{1+T} \right) \exp \left( - \frac{T \|x\|^{\alpha}}{  g(D) \snr  } \right)  \right],
\end{align}
where we have used the independence of $D$ and $F$, and $F \sim \exp (1)$. Therefore, $\tilde{\Phi}^{(\textrm{u})}$ is a PPP with intensity measure
$
\Lambda (A) = \int_A p (x) \lambda \dint x.
$
Further, $\Upsilon^{(\textrm{u})} = \tilde{\Phi}^{(\textrm{u})} (\b R^2)$ is Poisson with parameter 
\begin{align}
&\Lambda (\b R^2) = \int_{\b R^2} p (x) \lambda \dint x \notag \\
&= \int_{\b R^2}  \mathbb{E}_D \left[ \mathbb{I} \left( g(D) > \frac{T}{1+T} \right) \exp \left( - \frac{T \|x\|^{\alpha}}{  g(D) \snr  } \right)  \right] \lambda \dint x \notag \\
&= \pi \lambda \int_0^{\infty} \mathbb{E}_D \left[ \mathbb{I} \left( g(D) > \frac{T}{1+T} \right) \exp \left( - \frac{T v^{\frac{\alpha}{2}}}{  g(D) \snr  } \right)  \right]    \dint v  . \notag 
\end{align}

Next we show that $\Upsilon^{(\textrm{u})}$ can be truncated at  $\lfloor \frac{1+T}{T} \rfloor$, following a similar argument as in \cite{baccelli2009stochastic}. To this end, suppose there are $n$ decodable transmitters, without loss of generality assumed to be $X_0, ..., X_{n-1}$. Denoting by $\tilde{I} = I_{\Phi - \cup_{j=0}^{n-1} \delta_{X_j}}  + N_0 /E$, Then we have
\begin{align}
\frac{ g(D_i)   \|X_i\|^{-\alpha}  F_{i}  }{  -g(D_i)   \|X_i\|^{-\alpha}  F_{i}  + \sum_{j=0}^{n-1} \|X_j\|^{-\alpha}  F_{j} + \tilde{I}  } \geq T, 
\end{align}
for $i=0,...,n-1 $,
which implies that
\begin{align}
\frac{    \|X_i\|^{-\alpha}  F_{i}  }{  \sum_{j=0, j\neq i}^{n-1} \|X_j\|^{-\alpha}  F_{j} +  \tilde{I} } \geq T, 
\end{align}
for $i=0,...,n-1 $.
With some algebraic manipulations, we have the following set of inequalities:
\begin{align}
  (1+T)\|X_i\|^{-\alpha}  F_{i}  \geq T(  \sum_{j=0}^{n-1} \|X_j\|^{-\alpha}  F_{j} +  \tilde{I} ),  
\end{align}
for $i=0,...,n-1 $.
Summing the above set of inequalities,
\begin{align}
(1+T) \sum_{j=0}^{n-1} \|X_j\|^{-\alpha}  F_{j}  &\geq nT (\sum_{j=0}^{n-1} \|X_j\|^{-\alpha}  F_{j} +  \tilde{I} ) \notag \\
&> nT \sum_{j=0}^{n-1} \|X_j\|^{-\alpha}  F_{j} .
\end{align}
It follows that $n \leq \lfloor \frac{1+T}{T} \rfloor$, and thus the proposition has been proven.

\subsection{Proof of Proposition \ref{pro:2}}
\label{proof:pro:2}

To begin with, we condition on the location of the nearest transmitter $X_0 = x = (r,\theta)$ and its associated fading $F_0 = f$ and timing misalignment $D_0=\tau$. Then
\begin{align}
&\mathbb{P} ( \textrm{SINR}_0 \geq T | X_0 = x,  F_{0} = f, D_0 = \tau )  \notag \\
&= \mathbb{P}  \bigg( \frac{ g(\tau)   \|x\|^{-\alpha}  f  }{ (1- g(\tau))   \|x\|^{-\alpha}  f  + I_{\Phi- \delta_x}  + N_0 /E }   \geq T  \notag \\
& \quad \quad \quad  \quad  \quad \big| X_0 = x,  F_{0} = f, D_0 = \tau \bigg)     \notag \\
&= \mathbb{P} \bigg( f \geq r^\alpha h(\tau, T)( I_{\Phi- \delta_x}  + N_0 /E  ) \notag \\
& \quad \quad \quad  \quad  \quad \big| X_0 = x,  F_0 = f, D_0 = \tau \bigg) \mathbb{I} \left( g(\tau) > \frac{T}{1+T} \right)   \notag \\
&= \mathbb{P}^{x,f, \tau} \bigg( f \geq r^\alpha h(\tau, T) ( I_{\Phi - \delta_x} + N_0/E ) \notag \\
& \quad \quad \quad  \quad  \quad 
\big| \Phi ( B (o,r) ) = 0 \bigg ) \mathbb{I} \left( g(\tau) > \frac{T}{1+T} \right)  ,
\label{eq:9}
\end{align}
where $ \mathbb{P}^{x,f, \tau} (\cdot )$ denotes the Palm distribution with respect to $\Phi$, i.e., the probability law conditioned on that there exists a point at location $x$ with the marks $f$ and $\tau$. Note that, conditioned on that the nearest point is located in $x$, there are no other points in $\Phi$ located in the ball $B(o,r)$ centered at $o$ with radius $r$, i.e., $\Phi ( B (o,r) ) = 0$. This condition has been made explicitly  
in (\ref{eq:9}). Further, the first term in (\ref{eq:9}) equals
\begin{align}
&\mathbb{P}^{x,f,\tau} \bigg( f \geq r^\alpha h(\tau, T) ( I_{\Phi \cap B^c (o,r) - \delta_x} + N_0/E ) \notag \\
& \quad \quad \quad  \quad  \quad 
\big| \Phi ( B (o,r) ) = 0 \bigg)  \notag \\
&= \mathbb{P}^{x,f,\tau} ( f \geq r^\alpha h(\tau, T) ( I_{\Phi \cap B^c (o,r) - \delta_x} + N_0/E )   ) \notag \\
&= \mathbb{P} ( f \geq r^\alpha h(\tau, T) ( I_{\Phi \cap B^c (o,r) } + N_0/E )   ) .
\label{eq:10}
\end{align}
The first equality in (\ref{eq:10}) is due to the independence of $I_{\Phi \cap B^c (o,r) - \delta_x}$ and $\Phi ( B (o,r) ) = 0$, which follows from the complete independence property of PPP. The second equality in (\ref{eq:10}) is due to Slivnyak-Mecke Theorem \cite{baccelli2009stochastic}.

Following a similar derivation as in \cite{andrews2011tractable}, we can uncondition on $F_0 = f$ and $X_0=x$ to obtain
\begin{align}
&\mathbb{P} ( \textrm{SINR}_0 \geq T | D_0 = \tau )  
= \pi \lambda \mathbb{I} \left( g(\tau) > \frac{T}{1+T} \right)  \notag \\
& \times \int_0^\infty e^{ -v^{\frac{\alpha}{2}} h(\tau, T) N_0 /E  } e^{ -\pi \lambda v ( 1 + \rho (h(\tau, T) ,\alpha ) ) } \dint v  ,
\end{align}
where $\rho (t, \alpha)$ is defined in Prop. \ref{pro:2}. Unconditioning further on $D = \tau$ yields (\ref{eq:11}). This completes the proof.

\bibliographystyle{IEEEtran}
\bibliography{IEEEabrv,first}

\begin{thebibliography}{10}
\providecommand{\url}[1]{#1}
\csname url@samestyle\endcsname
\providecommand{\newblock}{\relax}
\providecommand{\bibinfo}[2]{#2}
\providecommand{\BIBentrySTDinterwordspacing}{\spaceskip=0pt\relax}
\providecommand{\BIBentryALTinterwordstretchfactor}{4}
\providecommand{\BIBentryALTinterwordspacing}{\spaceskip=\fontdimen2\font plus
\BIBentryALTinterwordstretchfactor\fontdimen3\font minus
  \fontdimen4\font\relax}
\providecommand{\BIBforeignlanguage}[2]{{%
\expandafter\ifx\csname l@#1\endcsname\relax
\typeout{** WARNING: IEEEtran.bst: No hyphenation pattern has been}%
\typeout{** loaded for the language `#1'. Using the pattern for}%
\typeout{** the default language instead.}%
\else
\language=\csname l@#1\endcsname
\fi
#2}}
\providecommand{\BIBdecl}{\relax}
\BIBdecl

\bibitem{lin2013overview}
X.~Lin, J.~G. Andrews, A.~Ghosh, and R.~Ratasuk, ``An overview of {3GPP}
  device-to-device proximity services,'' \emph{IEEE Communications Magazine},
  vol.~52, no.~4, pp. 40--48, April 2014.

\bibitem{vasudevan2005neighbor}
S.~Vasudevan, J.~Kurose, and D.~Towsley, ``On neighbor discovery in wireless
  networks with directional antennas,'' in \emph{Proceedings of IEEE INFOCOM},
  vol.~4, 2005, pp. 2502--2512.

\bibitem{keeler2013sinr}
H.~P. Keeler, B.~Blaszczyszyn, and M.~K. Karray, ``{SINR}-based k-coverage
  probability in cellular networks with arbitrary shadowing,'' in
  \emph{Proceedings of IEEE International Symposium on Information Theory
  (ISIT)}, 2013, pp. 1167--1171.

\bibitem{baccelli2009stochastic}
F.~Baccelli and B.~Blaszczyszyn, \emph{Stochastic Geometry and Wireless
  Networks - Part I: Theory}.\hskip 1em plus 0.5em minus 0.4em\relax Now
  Publishers Inc., 2009.

\bibitem{blaszczyszyn2013using}
B.~Blaszczyszyn, M.~K. Karray, and H.~P. Keeler, ``Using {Poisson} processes to
  model lattice cellular networks,'' in \emph{Proceedings of IEEE INFOCOM},
  2013, pp. 773--781.

\bibitem{GuoHae14}
A.~Guo and M.~Haenggi, ``Asymptotic deployment gain: A simple approach to
  characterize the {SINR} distribution in general cellular networks,''
  \emph{submitted}, April 2014. Available at http://arxiv.org/abs/1404.6556.

\bibitem{baccelli2006aloha}
F.~Baccelli, B.~Blaszczyszyn, and P.~Muhlethaler, ``An {Aloha} protocol for
  multihop mobile wireless networks,'' \emph{IEEE Transactions on Information
  Theory}, vol.~52, no.~2, pp. 421--436, February 2006.

\bibitem{andrews2011tractable}
J.~G. Andrews, F.~Baccelli, and R.~K. Ganti, ``A tractable approach to coverage
  and rate in cellular networks,'' \emph{IEEE Transactions on Communications},
  vol.~59, no.~11, pp. 3122--3134, November 2011.

\bibitem{wu2013flashlinq}
X.~Wu, S.~Tavildar, S.~Shakkottai, T.~Richardson, J.~Li, R.~Laroia, and
  A.~Jovicic, ``{FlashLinQ: A} synchronous distributed scheduler for
  peer-to-peer ad hoc networks,'' \emph{IEEE/ACM Transactions on Networking},
  vol.~21, no.~4, pp. 1215--1228, August 2013.

\bibitem{baccelli2012optimizing}
F.~Baccelli, J.~Li, T.~Richardson, S.~Shakkottai, S.~Subramanian, and X.~Wu,
  ``On optimizing {CSMA} for wide area ad hoc networks,'' \emph{Queueing
  Systems}, vol.~72, no. 1-2, pp. 31--68, October 2012.

\bibitem{erturk2013distributions}
M.~C. Erturk, S.~Mukherjee, H.~Ishii, and H.~Arslan, ``Distributions of
  transmit power and {SINR} in device-to-device networks,'' \emph{IEEE
  Communications Letters}, vol.~17, no.~2, pp. 273--276, February 2013.

\bibitem{lin2013multicast}
X.~Lin, R.~Ratasuk, A.~Ghosh, and J.~G. Andrews, ``Modeling, analysis and
  optimization of multicast device-to-device transmissions,'' \emph{IEEE
  Transactions on Wireless Communications}, vol.~13, no.~8, pp. 4346--4359,
  August 2014.

\bibitem{lin2013comprehensive}
X.~Lin, J.~G. Andrews, and A.~Ghosh, ``Spectrum sharing for device-to-device
  communication in cellular networks,'' \emph{IEEE Transactions on Wireless
  Communications}, to appear, 2014.

\bibitem{li2014message}
Y.~Li and W.~Wang, ``Message dissemination in intermittently connected {D2D}
  communication networks,'' \emph{IEEE Transactions on Wireless
  Communications}, vol.~13, no.~7, pp. 3978--3990, July 2014.

\bibitem{haenggi2009stochastic}
M.~Haenggi, J.~G. Andrews, F.~Baccelli, O.~Dousse, and M.~Franceschetti,
  ``Stochastic geometry and random graphs for the analysis and design of
  wireless networks,'' \emph{IEEE Journal on Selected Areas in Communications},
  vol.~27, no.~7, pp. 1029--1046, September 2009.

\bibitem{elsawy2013stochastic}
H.~ElSawy, E.~Hossain, and M.~Haenggi, ``Stochastic geometry for modeling,
  analysis, and design of multi-tier and cognitive cellular wireless networks:
  A survey,'' \emph{IEEE Communications Surveys \& Tutorials}, vol.~15, no.~3,
  pp. 996--1019, July 2013.

\bibitem{pollet1994ber}
T.~Pollet, P.~Spruyt, and M.~Moeneclaey, ``The {BER} performance of {OFDM}
  systems using non-synchronized sampling,'' in \emph{Proceedings of IEEE
  GLOBECOM}, 1994, pp. 253--257.

\bibitem{schmidl1997robust}
T.~M. Schmidl and D.~C. Cox, ``Robust frequency and timing synchronization for
  {OFDM},'' \emph{IEEE Transactions on Communications}, vol.~45, no.~12, pp.
  1613--1621, December 1997.

\bibitem{speth1999optimum}
M.~Speth, S.~A. Fechtel, G.~Fock, and H.~Meyr, ``Optimum receiver design for
  wireless broad-band systems using {OFDM - Part I},'' \emph{IEEE Transactions
  on Communications}, vol.~47, no.~11, pp. 1668--1677, November 1999.

\bibitem{wang2003ser}
X.~Wang, T.~T. Tjhung, Y.~Wu, and B.~Caron, ``{SER} performance evaluation and
  optimization of {OFDM} system with residual frequency and timing offsets from
  imperfect synchronization,'' \emph{IEEE Transactions on Broadcasting},
  vol.~49, no.~2, pp. 170--177, June 2003.

\bibitem{mostofi2006mathematical}
Y.~Mostofi and D.~C. Cox, ``Mathematical analysis of the impact of timing
  synchronization errors on the performance of an {OFDM} system,'' \emph{IEEE
  Transactions on Communications}, vol.~54, no.~2, pp. 226--230, February 2006.

\bibitem{tonello2000analysis}
A.~M. Tonello, N.~Laurenti, and S.~Pupolin, ``Analysis of the uplink of an
  asynchronous multi-user {DMT OFDMA} system impaired by time offsets,
  frequency offsets, and multi-path fading,'' in \emph{Proceedings of IEEE
  Vehicular Technology Conference (VTC)}, vol.~3, 2000, pp. 1094--1099.

\bibitem{el2001ofdm}
M.~S. El-Tanany, Y.~Wu, and L.~H{\'a}zy, ``{OFDM} uplink for interactive
  broadband wireless: {Analysis} and simulation in the presence of carrier,
  clock and timing errors,'' \emph{IEEE Transactions on Broadcasting}, vol.~47,
  no.~1, pp. 3--19, March 2001.

\bibitem{park2003performance}
M.~Park, K.~Ko, H.~Yoo, and D.~Hong, ``Performance analysis of {OFDMA} uplink
  systems with symbol timing misalignment,'' \emph{IEEE Communications
  Letters}, vol.~7, no.~8, pp. 376--378, August 2003.

\bibitem{raghunath2009sir}
K.~Raghunath and A.~Chockalingam, ``{SIR} analysis and interference
  cancellation in uplink {OFDMA} with large carrier frequency/timing offsets,''
  \emph{IEEE Transactions on Wireless Communications}, vol.~8, no.~5, pp.
  2202--2208, May 2009.

\bibitem{hamdi2009interference}
K.~A. Hamdi and Y.~M. Shobowale, ``Interference analysis in downlink {OFDM}
  considering imperfect intercell synchronization,'' \emph{IEEE Transactions on
  Vehicular Technology}, vol.~58, no.~7, pp. 3283--3291, September 2009.

\bibitem{medjahdi2011performance}
Y.~Medjahdi, M.~Terr{\'e}, D.~Le~Ruyet, D.~Roviras, and A.~Dziri, ``Performance
  analysis in the downlink of asynchronous {OFDM/FBMC} based multi-cellular
  networks,'' \emph{IEEE Transactions on Wireless Communications}, vol.~10,
  no.~8, pp. 2630--2639, August 2011.

\bibitem{hamdi2010precise}
K.~A. Hamdi, ``Precise interference analysis of {OFDMA} time-asynchronous
  wireless ad-hoc networks,'' \emph{IEEE Transactions on Wireless
  Communications}, vol.~9, no.~1, pp. 134--144, January 2010.

\bibitem{3gppD2dRF}
{3GPP}, ``3rd generation partnership project; technical specification group
  radio access network; study on {LTE} device to device proximity services;
  radio aspects (release 12),'' \emph{TR 36.843 V12.0.1}, March 2014.

\bibitem{tuomaala2005effective}
E.~Tuomaala and H.~Wang, ``Effective {SINR} approach of link to system mapping
  in {OFDM}/multi-carrier mobile network,'' in \emph{International Conference
  on Mobile Technology, Applications and Systems}, 2005, pp. 1--5.

\bibitem{haenggi2009interference}
M.~Haenggi and R.~K. Ganti, \emph{Interference in large wireless
  networks}.\hskip 1em plus 0.5em minus 0.4em\relax Now Publishers Inc., 2009.

\end{thebibliography}

\end{document}